\documentclass{article}
\usepackage{ijcai26}

\usepackage{times}
\usepackage{soul}
\usepackage{url}
\usepackage[hidelinks]{hyperref}
\usepackage[utf8]{inputenc}
\usepackage[small]{caption}
\usepackage{graphicx}
\usepackage{amsmath}
\usepackage{amsthm}
\usepackage{booktabs}
\usepackage[vlined,linesnumbered,ruled]{algorithm2e}
\usepackage[switch]{lineno}
\usepackage{algorithmic}
\usepackage{xcolor}

\SetNlSty{scriptsize}{}{}

\newtheorem{example}{Example}

\def\gdmc{{\tt gDMC}}
\title{\gdmc: A Generic Distributed Model Counting Framework via Work-Stealing}

\author{
  Zhenghang Xu$^{1, 2}$\and
  Minghao Yin$^{1, 2}$ \and
  Junping Zhou$^{1, 2}$ \and
  Jean-Marie Lagniez$^3$ \\
  \affiliations
  $^1$School of Information Science and Technology, Northeast Normal University, Changchun, China \\
  $^2$Key Laboratory of Applied Statistics of MOE, Northeast Normal University, Changchun, China \\
  $^3$Univ. Artois, CNRS, CRIL, F-62300 Lens, France\\
  \emails
  \{xuzh121, ymh, zhoujp877\}@nenu.edu.cn,
  lagniez@cril.fr
}

\def\cnf{{CNF}}

\def\discount{{\tt DisCount}}
\def\sharpsat{{\tt SharpSAT-TD}}
\def\d4{{\tt D4}}
\def\countAntom{{\tt CountAntom}}
\def\dcountAntom{{\tt dCountAntom}}
\def\dmc{{\tt DMC}}
\def\ganak{{\tt GANAK}}

\def\bcp{{\tt BCP}}
\newcommand{\var}[1]{Var (#1)\xspace}
\newcommand{\norm}[1]{\left\|#1\right\|}

\newtheorem{definition}{Definition}
\newtheorem{proposition}{Proposition}

\begin{document}

\maketitle

\begin{abstract}
Propositional Model Counting ($\#\mathsf{SAT}$) is essential for probabilistic reasoning but faces scalability limits on single cores. Existing distributed approaches struggle with high initialization overheads (static decomposition) or rigid architecture. We propose a novel, generic framework for distributed \emph{exact} model counting. Leveraging C++ templates, our architecture decouples parallel orchestration from solving logic, enabling state-of-the-art solvers to be parallelized with minimal modification. We implement an adaptive work-stealing strategy that ensures effective load balancing. Experiments on competition benchmarks show that our approach achieves near-linear scalability and significantly outperforms existing distributed solvers.
\end{abstract}

\section{Introduction}
Model counting (MC), also known as $\#\mathsf{SAT}$, is the problem of computing the number of models (satisfying assignments) of a given propositional formula, typically in \cnf. 
Its direct extension, \emph{weighted model counting} (WMC), is of tremendous importance in a wide range of AI applications, including probabilistic inference~\cite{DBLP:conf/aaai/SangBK05,DBLP:journals/ai/ChaviraD08}, planning under uncertainty~\cite{DBLP:conf/aips/DomshlakH06}, and neural network verification~\cite{DBLP:conf/ccs/BalutaSSMS19}. 
Beyond AI, model counting is critical in domains such as hardware testing and reliability estimation~\cite{DBLP:conf/dft/FeitenSSCBPB12,DBLP:conf/qest/TeuberW21,DBLP:conf/aaai/Duenas-OsorioMP17,DBLP:conf/cav/GirolFB21,DBLP:conf/cav/MeiBL24}.

The computational hardness of model counting is underscored by Toda's theorem~\cite{DBLP:journals/siamcomp/Toda91}, which establishes that $\mathsf{P}^{\#\mathsf{P}}$ encompasses the entire polynomial hierarchy (PH).
Despite this theoretical hardness, significant progress has been made in the design of efficient sequential model counters. 
Modern search-based solvers, such as \d4~\cite{DBLP:conf/ijcai/LagniezM17}, \ganak~\cite{DBLP:conf/ijcai/SharmaRSM19}, and \sharpsat~\cite{DBLP:conf/cp/KorhonenJ21}, leverage Conflict-Driven Clause Learning (CDCL), component caching, and implicit Boolean Constraint Propagation (BCP) to prune the search space effectively. 
However, sequential performance is increasingly constrained by the physical limits of single-core clock speeds and the memory required to maintain massive component caches.

The importance of $\#\mathsf{SAT}$ has driven significant research into scalable algorithms, evidenced by recent model counting competitions~\cite{Competition2021_23}.
To handle the growing complexity of real-world instances, parallel and distributed computing offers a promising path forward. 
However, distributing $\#\mathsf{SAT}$ is fundamentally harder than distributing $\mathsf{SAT}$. 
In $\mathsf{SAT}$, the search can stop as soon as \emph{one} solution is found, allowing for aggressive, speculative strategies. 
In $\#\mathsf{SAT}$, the entire solution space must be covered, meaning the total runtime is determined by the slowest worker (the `straggler' problem).
While GPU-accelerated approaches exist~\cite{DBLP:conf/cp/FichteHR21}, they represent a distinct architectural paradigm; this paper focuses on distributed CPU-based solving.

State-of-the-art parallel and distributed solvers, such as \dcountAntom~\cite{DBLP:conf/cluster/BurchardSB16} and the more recent \discount~\cite{DBLP:conf/kr/LagniezL25}, typically rely on a \emph{Cube-and-Conquer} strategy. 
These methods decompose the search space into independent subproblems (cubes) which are then distributed across nodes.
However, this approach has several limitations. First, the initialization overhead is significant; as noted in~\cite{DBLP:conf/kr/LagniezL25}, the decomposition phase can be costly enough to negate performance gains on medium-hard instances. 
Second, static decomposition suffers from \emph{redundancy}: separate workers may unknowingly explore identical search spaces. For instance, when different cubes lead to the same connected component, resulting in duplicated effort.
Finally, achieving effective load balancing requires generating a massive number of cubes, a parameter that is notoriously difficult to tune and highly instance-dependent.

The distributed solver \dmc~\cite{DBLP:conf/ijcai/LagniezMS18} offers a compelling alternative to static partitioning by employing a dynamic work-stealing mechanism. 
Unlike \dcountAntom{}, \dmc{} avoids the exchange of formula fragments and learned clauses, significantly minimizing communication overhead. 
However, its implementation is tightly coupled to the \d4{} solver, requiring invasive modifications to the sequential core that hinder modularity and maintenance. 
Additionally, \dmc{} uses double-precision floating-point arithmetic for aggregation; this induces precision loss on instances with large model counts, rendering it unsuitable for exact counting tasks.

To overcome these limitations, we propose a novel C++ framework that leverages \emph{generic programming} (via Concepts and Templates) to enable efficient, solver-agnostic distributed model counting within the class of DPLL-style counters. Our architecture employs a modular stack that separates distributed orchestration from core solving logic, enabling state-of-the-art solvers to be parallelized with minimal modification. By resolving communication types at compile time, this design ensures zero-cost abstractions with negligible runtime overhead. Furthermore, we formally establish the conditions necessary for safe and complete work sharing.

Specifically, our contributions are threefold.
First, we provide a generic, template-based architecture that simplifies the integration of work-stealing into existing DPLL-style model counters. Second, we identify the formal requirements for valid work transfer, establishing the theoretical guarantees necessary for the soundness and completeness of the distributed search. Finally, we demonstrate through extensive experiments that our solution outperforms both static partitioning (e.g., \discount) and existing work-stealing solvers (e.g., \dmc), achieving near-linear speedups and solving significantly more instances.

\section{Preliminaries~\label{formal}}
Let $\mathcal{L}$ represent the propositional language formed from a finite set $\mathcal{P}$ of propositional variables, using the standard logical connectives ($\neg$, $\vee$, $\wedge$) and the Boolean constants $\top$ (true) and $\bot$ (false).
A \emph{literal} $\ell$ is either a variable $x \in \mathcal{P}$ or its negation $\bar{x}$. 
A \emph{cube} (or term) is a conjunction of literals, and a \emph{clause} is a disjunction of literals. 
When convenient, clauses and cubes are regarded as sets of literals. 
A \emph{CNF formula} is a conjunction of clauses, also viewed as a set of clauses.

Formulas are interpreted in the classical way: an \emph{interpretation} $\omega$ is a mapping from $\mathcal{P}$ to $\{0, 1\}$. 
An interpretation $\omega$ is a \emph{model} of a formula $\Sigma$ if $\Sigma$ evaluates to $1$ under $\omega$ (denoted $\omega \models \Sigma$). 
We use $\models$ to denote logical entailment and $\equiv$ for logical equivalence.
The notation $\norm{\Sigma}$ denotes the number of models (i.e., satisfying assignments) of the formula $\Sigma$ over its variables $\var{\Sigma}$.
The \emph{\#SAT problem} asks to compute $\norm{\Sigma}$ for a given formula and is the canonical example of a $\#\mathsf{P}$-complete problem.

\begin{example}\label{ex:running}
    As a running example, let us consider the CNF formula $\Sigma$ defined over $\var{\Sigma} = \{x_1, \ldots, x_7\}$:
    \normalsize{
    \begin{alignat*}{4}
            \bar{x_1} \vee \bar{x_2} \vee \bar{x_6}    \quad && \bar{x_1} \vee       x_4 \vee x_5       \quad && x_1       \vee x_4 \vee x_5 \quad && x_1 \vee x_3\\
            \bar{x_1} \vee          \bar{x_2} \vee x_6 \quad && \bar{x_1} \vee       x_2 \vee x_3       \quad && \bar{x_1} \vee x_2 \vee x_3 \vee \bar{x_7}  \\
            \bar{x_1} \vee          \bar{x_3} \vee x_6 \quad && \bar{x_1} \vee \bar{x_3} \vee \bar{x_6} \quad && \bar{x_1} \vee      x_3 \vee x_4 \vee x_7
\end{alignat*}
    }
    In this case, $\norm{\Sigma} = 24$.    
\end{example}

The Boolean Constraint Propagator, $\bcp(\Sigma)$, performs \emph{unit propagation} by iteratively enforcing unit clauses within the formula $\Sigma$.
$\bcp(\Sigma)$ returns $\bot$ (conflict) if unit propagation derives the empty clause (a unit refutation). 
Otherwise, it returns the set of literals $S$ derived from $\Sigma$.
The \emph{conditioning} of $\Sigma$ by a literal $\ell$, denoted $\Sigma|_{\ell}$, is obtained by replacing every occurrence of $\ell$ with $\top$ and every occurrence of its negation $\bar \ell$ with $\bot$.
This is followed by standard Boolean simplifications to remove satisfied clauses and removing $\bot$ literals from the remaining clauses until the formula is minimal.
This notion extends to a set of literals $S = \{\ell_1, \ldots, \ell_m\}$ by iterative application: $\Sigma|_{S} = (\dots(\Sigma|_{\ell_1})\dots)|_{\ell_m}$.

\begin{example}[Example~\ref{ex:running} cont'ed]
    Assume we branch on $x_7$ and $x_1$. Then $\bcp(\Sigma \wedge x_7 \wedge x_1) = \{x_7, x_1\}$. Conditioning yields:
$\Sigma|_{\{x_7,x_1\}} = \{ 
    (x_2 \vee  x_3), (x_4 \vee  x_5), (\bar{x_2} \vee  x_6),
    (\bar{x_3} \vee \bar{x_6}), (\bar{x_2} \vee \bar{x_6}), (\bar{x_3} \vee  x_6)\}
    $.
    Branching on the literals $\{x_7, \bar{x_1}\}$ yields $\bcp(\Sigma \wedge x_7 \wedge \bar{x_1}) = \{x_7, \bar{x_1}, x_3\}$ and $\Sigma|_{\{x_7,\bar{x_1},x_3\}} = \{ (x_4 \vee x_5) \}$.
\end{example}

State-of-the-art exact model counters typically extend the CDCL architecture used in SAT solving. 
Unlike SAT solvers, which terminate upon finding a single model, model counters must traverse the entire solution space. 
To do this efficiently, they rely on three key techniques: \textbf{Clause Learning}, where the solver learns a new clause upon conflict to prune future search; \textbf{Component Caching}, which stores results of previously solved subproblems to avoid redundant computations~\cite{DBLP:conf/sat/SangBBKP04}; and \textbf{Component Decomposition}, which exploits independent sub-formulas $\Sigma_1, \dots, \Sigma_k$ by computing the product of their counts $\prod_{i=1}^k \norm{\Sigma_i}$.

The standard model counting procedure generally follows the structure of Algorithm~\ref{alg:count} (presented in Section~\ref{sec:stack}, ignoring the modifications highlighted in blue and red).
First, $\bcp(\Sigma)$ is executed to gather the set of implied literals $S$, and the formula is simplified to $\Sigma' = \Sigma|_S$.
If a conflict arises (i.e., $\Sigma' = \bot$), the function returns $0$.
If all clauses are satisfied (i.e., $\Sigma' = \emptyset$), the algorithm returns $2^{|\var{\Sigma}| - |S|}$, which accounts for the single model of the implied literals and the $2^k$ combinations of any remaining variables.
Otherwise, the algorithm checks the cache. 
If the instance is new, it decomposes the formula into connected components. 
If multiple components exist, they are solved recursively, and their results are multiplied. 
If the formula forms a single component, the algorithm selects a branching variable $x$, computes the sum of the counts for the two branches ($\Sigma' \wedge x$ and $\Sigma' \wedge \bar x$), and caches the result.
Finally, in both the component and branching cases, the computed value $val$ is scaled by $2^{|\var{\Sigma}| - |\var{\Sigma'}| - |S|}$. 
This factor corrects for variables that were present in $\Sigma$ but eliminated during simplification (e.g., variables that became free), ensuring the count remains correct with respect to the original variable set.

Crucially, the interaction between dynamic clause learning and component caching requires careful management.
When a conflict occurs, a learned clause is added to the formula. 
Although this clause is a logical consequence of the original formula, its presence modifies the implication landscape. 
As established in~\cite{DBLP:conf/sat/SangBBKP04}, naively retrieving cached components in the presence of different learned clauses can lead to incorrect counts.
We explicitly address this constraint in Section~\ref{gdmc}, where we define the necessary conditions for safe job sharing in a distributed setting.

The efficiency of the solver heavily relies on the variable selection heuristic ($\mathsf{chooseVariable}$). 
Modern approaches exploit the formula's \emph{primal graph} and its \emph{tree decomposition} $(T, \{B_t\})$, a mapping of variables into a tree structure satisfying standard connectivity properties~\cite{DBLP:journals/jal/RobertsonS86}.
Although real-world instances often have high \emph{treewidth} (defined as $\max |B_t| - 1$), heuristics that prioritize variables appearing in the ``top'' bags of the decomposition significantly improve performance by fostering component decomposition~\cite{DBLP:conf/cp/KorhonenJ21}.

\section{The Explicit Stack Framework~\label{sec:stack}}
Exact model counting traverses an \emph{AND-OR Tree}, summing results at \textbf{OR-Nodes} (decisions) and multiplying at \textbf{AND-Nodes} (decompositions). 
In standard DFS solvers, pending work remains implicit within the recursion stack. 
This implicit representation creates a major barrier for distribution, as recursion frames cannot be easily ``detached'' to offload tasks to other workers.

\subsection{The Explicit Search State}

To enable flexible \emph{work-stealing}, \gdmc{} reifies the recursion stack into a manageable data structure: the \textbf{Explicit Stack} $\mathcal{S}$. 
Crucially, we do not materialize the entire tree (which grows exponentially); we maintain only the \emph{current active path}. 
Any node on this path with an unexplored sibling represents a potential \textbf{job} that can be offloaded.

It is important to note that $\mathcal{S}$ does not behave strictly as a standard Last-In-First-Out (LIFO) stack during distributed execution. 
In a purely sequential setting, a node is popped as soon as its local sub-problems are solved. 
In \gdmc{}, however, a node $\nu$ may have offloaded portions of its work to other solvers. 
Consequently, the node serves a dual purpose: it is both a controller for local iteration and a dependency tracker for remote jobs. 
Even if the local set of tasks is empty, the node must persist in the stack if it is still waiting for remote results. 
It acts as an aggregation point, waiting to combine the local accumulator with the results of the outstanding shared tasks. 
Only when all local and remote obligations are met is the node considered fully resolved and removed from $\mathcal{S}$.

\begin{definition}[Stack Nodes]
The stack $\mathcal{S}$ is a sequence of nodes $\nu = \langle \mathcal{F}, v, \mathcal{I}, \mu \rangle$ comprising:
\begin{itemize}
    \item $\mathcal{F}$: The set of pending sub-problems to be solved locally.
    \item $v$: The local accumulator for the model count.
    \item $\mathcal{I}$: The set of identifiers for sub-problems offloaded.
    \item $\mu$: A local scaling factor (representing variables eliminated at this level), always initialized to $1$.
\end{itemize}
We distinguish two node types based on their initialization and aggregation logic. A \textbf{Decision Node (Sum)} is initialized with branches $\{\Sigma \wedge x, \Sigma \wedge \bar{x}\}$, where the accumulator $v$ starts at $0$ as the additive identity. Conversely, a \textbf{Decomposition Node (Product)} is initialized with disjoint components $\{\Sigma_1, \dots, \Sigma_k\}$ and an accumulator $v$ starting at $1$, the identity for multiplication.
\end{definition}

\begin{example}[Ex.~\ref{ex:running} cont'ed]\label{ex:stack}
Branching on $x_7, x_1$ partitions the formula into components $\mathcal{C}_1 = \{x_4 \vee x_5\}$ and $\mathcal{C}_2 = \{(x_2 \vee x_3),\allowbreak (\bar{x_2} \vee x_6),\allowbreak (\bar{x_3} \vee \bar{x_6}),\allowbreak (\bar{x_2} \vee \bar{x_6}),\allowbreak (\bar{x_3} \vee x_6)\}$.
The stack is:
$\mathcal{S} = \big( 
    \langle \{\Sigma|_{\bar{x_7}}\}, 0, \emptyset, 1 \rangle, \allowbreak
    \langle \{\Sigma|_{x_7, \bar{x_1}}\}, 0, \emptyset, 1 \rangle, \allowbreak
    \langle \{\mathcal{C}_1, \mathcal{C}_2\}, 1, \emptyset, 1 \rangle \big)$.
\end{example}

While the formal definition describes nodes as containing explicit formulas $\mathcal{F}$, materializing these is inefficient in practice. Instead, \gdmc{} leverages the fact that any search tree sub-problem is uniquely determined by its path from the root. Specifically, a task $\Sigma_{sub}$ is represented implicitly by: the original formula $\Sigma_{root}$, the current \textbf{partial assignment} (unit literals from decisions and propagations), and the set of \textbf{relevant variables} restricted to the current connected component.

Consequently, when we describe ``pushing a formula $\Sigma \wedge x$'' onto the stack, the implementation simply records the decision literal $x$, the set of propagated unit literals, and the variables identifying the current component. 
The solver reconstructs the effective formula by applying the current assignment to $\Sigma_{root}$ and filtering by the component's variables. 
Similarly, when a job is offloaded to a remote worker, we transmit only the sequence of decisions (the ``cube'') that defines the path to that node, along with the set of variables for the component, allowing the remote worker to reconstruct the exact same state locally.

\subsection{The Challenge of Distributed Caching}

Integrating work-stealing into model counting is complicated by \textbf{Component Caching}. The standard caching invariant requires that if a formula $\Sigma$ is solved, its count $\norm{\Sigma}$ is stored for immediate retrieval. However, offloading $\Sigma$ to a remote worker creates a \emph{broken dependency}: the local solver cannot wait for the result without losing parallelism, nor can it cache a ``future'' result without prohibitive overhead. This ``hole'' propagates upward, preventing all parent nodes from being cached.

To mitigate this, \gdmc{} restricts sharing to the \textbf{root of the stack} (the least recently generated sub-problems). This topological constraint ensures the active frontier at the top of the stack remains entirely local and free of remote dependencies, allowing sub-trees to be computed and cached synchronously. Furthermore, anchoring sharing to the stack base preserves a low-overhead linear structure, avoiding the need to manage a complex AND-OR dependency tree in memory.

\subsection{Stack Primitives}

As $\mathcal{S}$ retains persistent nodes awaiting remote results, the logical frontier diverges from the physical top. We bridge this by maintaining a pointer $pos$ to the active context, denoted $\nu_{curr} = \mathcal{S}[pos] = \langle \mathcal{F}, v, \mathcal{I}, \mu \rangle$.

\begin{itemize}
    \item $\mathcal{S}.\mathsf{pushOr}(\Sigma_1, \Sigma_2)$: Pushes a new \textbf{Decision Node} onto $\mathcal{S}$. 
    This operation is valid only if $pos$ corresponds to the physical top of the stack. 
    We increment $pos$ and initialize the new tuple as $\langle \{\Sigma_1, \Sigma_2\}, 0, \emptyset, 1 \rangle$.
    
    \item $\mathcal{S}.\mathsf{pushAnd}(\{\Sigma_1, \dots, \Sigma_k\})$: Pushes a new \textbf{Decomposition Node} onto $\mathcal{S}$. 
    This operation is valid only if $pos$ refers to the physical top of the stack. 
    We increment $pos$ and initialize the tuple as $\langle \{\Sigma_1, \dots, \Sigma_k\}, 1, \emptyset, 1 \rangle$.

    \item $\mathcal{S}.\mathsf{alreadyShared}(\Sigma)$: Checks if the sub-problem $\Sigma$ is present in the pending formula set $\mathcal{F}$ of $\nu_{curr}$.
    Returns $\mathtt{true}$ if $\Sigma \notin \mathcal{F}$ (implying $\Sigma$ has been stolen or removed by a concurrent worker).
    Returns $\mathtt{false}$ otherwise.
    
    \item $\mathcal{S}.\mathsf{process}(\Sigma)$: Begins local processing of $\Sigma$ by removing it from $\mathcal{F}$ in $\nu_{curr}$. If $\mathcal{F}$ becomes empty, local work for the node is complete. If the remote identifier set $\mathcal{I}$ is also empty, the node is fully resolved and popped from $\mathcal{S}$ (with $pos$ decremented). Otherwise, the node persists as a dependency tracker for remote results, and we only decrement $pos$ to return control to the parent context.

   \item $\mathcal{S}.\mathsf{identity}()$: Returns the neutral element for the operation associated with $\nu_{curr}$ ($0$ if $\nu_{curr}$ is a Decision Node, $1$ if it is a Decomposition Node).
    
    \item $\mathcal{S}.\mathsf{isShared}()$: Checks if the current node $\nu_{curr}$ has partially offloaded its work. 
    Returns $\mathtt{true}$ if the set of remote identifiers $\mathcal{I}$ is not empty.
        
    \item $\mathcal{S}.\mathsf{complete}(c \times \mu)$: Updates the accumulator $v$ of $\nu_{curr}$ with the result of a sub-computation.
    The value to aggregate is derived from $c$ (the raw count computed by the solver) and $\mu$ (the scaling factor, typically $2^k$, for variables free in the sub-problem).    
\end{itemize}

\subsection{Integration in the Counting Algorithm}

Algorithm~\ref{alg:count} illustrates the integration of the \gdmc{} framework into a standard DPLL-style model counter. 
The implementation utilizes two global controllers: the explicit stack manager $\mathcal{S}$, which maintains the search state and aggregates results, and the communication manager  (referenced as \texttt{co}), which handles network protocols and work redistribution. 
The core counting logic (comprising unit propagation, component analysis, and variable branching) is shown in black and adheres to the classical formulation. 
The \textbf{\textcolor{blue}{blue blocks}} mark the synchronization points with $\mathcal{S}$, while the \textbf{\textcolor{red}{red block}} orchestrates the interaction with remote workers to enable dynamic load balancing.

The first interaction occurs immediately upon entering the function. 
Before performing any propagation, the solver queries $\mathcal{S}.\mathsf{alreadyShared}(\Sigma)$. 
This acts as a guard clause: if the current sub-problem $\Sigma$ has been removed from the stack (stolen by another worker), the function aborts immediately by returning the neutral element $\mathcal{S}.\mathsf{identity}()$. 
This prevents the local worker from solving a task that is no longer its responsibility. 
If the task remains, $\mathcal{S}.\mathsf{process}(\Sigma)$ is called to mark $\Sigma$ as ``busy,'' ensuring it cannot be stolen while the local worker is actively solving it.

Whenever the solver decomposes the problem, it registers the new sub-problems in the stack to expose the frontier of work. 
Specifically, before iterating over a set of connected components, the algorithm calls $\mathcal{S}.\mathsf{pushAnd}$ to create a new decomposition node containing these components. 
Similarly, prior to making a decision on a variable $x$, $\mathcal{S}.\mathsf{pushOr}$ is invoked to record the two resulting branches ($\Sigma \wedge x$ and $\Sigma \wedge \bar x$). 
These operations ensure that the explicit stack always reflects the current available tasks for potential thieves.

After the recursive calls return, the local solver has computed a partial result $val$. 
However, before returning this value, it must query $\mathcal{S}.\mathsf{isShared}()$. 
If this check returns $\mathtt{true}$, it indicates that sibling tasks were offloaded and the local result is incomplete. 
In this case, the algorithm calls $\mathcal{S}.\mathsf{complete}$ to aggregate $val$ into the persistent stack node and returns $\mathcal{S}.\mathsf{identity}()$, ensuring that the parent recursive call aggregates only the neutral element. 
Conversely, if the node was not shared, it implies the sub-problem was processed entirely locally. 
The stack node is effectively popped, and the computed $val$ is returned normally to the parent to continue the standard aggregation.

We defer the discussion of the \textbf{\textcolor{red}{red block}} in Algorithm~\ref{alg:count} to Section~4, where we detail the Master-Worker coordination and the network protocols required to interconnect individual solver instances.

\SetInd{0.5em}{0.5em}
\begin{algorithm}[t]
    \small
\caption{$\mathsf{count}(\Sigma)$~\label{alg:count}}
\SetKwInOut{Input}{Input}
\SetKwInOut{Output}{Output}

\Input{A CNF formula $\Sigma$}
\Output{The number of models $\norm{\Sigma}$}

\BlankLine{}

{\color{red}
\If{$\mathcal{S}.\mathsf{shouldPoll()}$ {\bf and} $\mathsf{co.shareRequest()}$}
{
$\mathsf{co.transferWorkToThiefs}(\mathcal{S})$\;
}
}

{\color{blue}
\lIf{$\mathcal{S}.\mathsf{alreadyShared}(\Sigma)$}{\Return{$\mathcal{S}.\mathsf{identity()}$}}
$\mathcal{S}.\mathsf{process}(\Sigma)$\;
}

$S \leftarrow \bcp(\Sigma); \quad \Sigma' \leftarrow \Sigma|_{S}$\;
\lIf{$\Sigma' = \bot$}{\Return{} $0$}
\lIf{$\Sigma' = \emptyset$}{\Return{} $2^{|\var{\Sigma}| - |S|}$}
\lIf{$\mathsf{inCache}(\Sigma')$}{\Return{} $\mathsf{getCache}(\Sigma')$}

$Components \leftarrow \mathsf{getConnectedComponents}(\Sigma')$\;

\If{$|Components| > 1$}{
    {\color{blue}
    $\mathcal{S}.\mathsf{pushAnd}(Components)$\;
    }

    $val \leftarrow 1$\;
    \lFor{$C \in Components$}{
        $val \leftarrow val \times \mathsf{count}(C)$
    }

    {\color{blue}
    \If{$\mathcal{S}.\mathsf{isShared()}$}
    {
        $\mathcal{S}.\mathsf{complete}(val \times 2^{|\var{\Sigma}| - |\var{\Sigma'}| - |S|})$\;
        \Return{$\mathcal{S}.\mathsf{identity()}$}\;
    }
    }

    \Return{} $val \times 2^{|\var{\Sigma}| - |\var{\Sigma'}| - |S|}$\;
}

$x \leftarrow \mathsf{chooseVariable}(\Sigma')$\;

{\color{blue}
    $\mathcal{S}.\mathsf{pushOr}(\Sigma' \wedge x, \Sigma' \wedge \bar x)$\;
}

$val \leftarrow \mathsf{count}(\Sigma' \wedge x) + \mathsf{count}(\Sigma' \wedge \bar x)$\;

{\color{blue}
    \If{$\mathcal{S}.\mathsf{isShared()}$}
    {
        $\mathcal{S}.\mathsf{complete}(val \times 2^{|\var{\Sigma}| - |\var{\Sigma'}| - |S|})$\;
        \Return{$\mathcal{S}.\mathsf{identity()}$}\;
    }
}

$\mathsf{addToCache}(\Sigma', val)$\;
\Return{} $val \times 2^{|\var{\Sigma}| - |\var{\Sigma'}| - |S|}$\;
\end{algorithm}

\section{Distributed Orchestration~\label{gdmc}}
\begin{algorithm}[t]
    \small
\caption{\gdmc{} Global Orchestration~\label{alg:gdmc}}
\SetKwInOut{Input}{Input}
\SetKwInOut{Output}{Output}

\Input{A CNF formula $\Sigma$ and a set of workers $\mathcal{W}$}
\Output{The number of models $\norm{\Sigma}$}

\BlankLine{}
\tcp{Initialization Phase}
$\Sigma \gets \mathtt{preprocess}(\Sigma)$\;
$\mathtt{broadcast}(\Sigma, \mathcal{W})$\;
$\mathtt{coordinateHeuristics}(\mathcal{W})$\;

\BlankLine{}
\tcp{Distributed Counting}
Select root worker $w_{r} \in \mathcal{W}$\;
$\mathtt{sendJob}(\Sigma, w_{r}, 0)$; $\quad idTask \gets 1$\;
$\mathit{idle} \gets \mathcal{W} \setminus \{w_{r}\}$; $\quad \mathit{contacted} \gets \mathsf{nil}$\;

\While{$\mathit{idle} \neq \mathcal{W}$}{
    \tcp{Select work source}
    \If{$\mathit{contacted} = \mathsf{nil}$ \textbf{and} $\mathit{idle} \neq \mathcal{W}$}{
        $\mathit{contacted} \gets \text{select one from } (\mathcal{W} \setminus \mathit{idle})$\;
        $\mathtt{requestWork}(\mathit{contacted})$\;
    }

    \tcp{Offload tasks to idle pool}
    \While{$\mathtt{isResponding}(\mathit{contacted})$ \textbf{and} $\mathtt{hasWork}(\mathit{contacted})$ \textbf{and} $\mathit{idle} \neq \emptyset$}
    {
        $w_{idle} \gets \mathtt{pop}(\mathit{idle})$\;
        $\mathtt{pairWorkers}(\mathit{contacted}, w_{idle}, idTask\text{++})$\;
    }

    \tcp{Update worker states}
    \ForEach{$w \in \mathcal{W} \setminus \mathit{idle}$}{
        \If{$\mathtt{isFinished}(w)$}{
            $\mathit{idle} \gets \mathit{idle} \cup \{w\}$\;
            \lIf{$w = \mathit{contacted}$}{
                $\mathit{contacted} \gets \mathsf{nil}$
            }
        }
    }
}

\BlankLine{}
\Return{$\mathtt{reduceResults}(\mathcal{W})$}\;
\end{algorithm}

Algorithm~\ref{alg:gdmc} employs a centralized master-worker architecture.
Initialization (lines 1--3) follows the \texttt{discount} protocol~\cite{DBLP:conf/kr/LagniezL25}: the Master broadcasts $\Sigma$ and enforces the cluster's best tree decomposition as a static branching order. The root task ($idTask = 0$) is then assigned to a single worker (line 5). 
Once initialized, the full problem is assigned to a single root worker via a blocking $\mathtt{sendJob}$ call with an initial task identifier $idTask = 0$ (line 5), while remaining workers enter an \emph{idle} state. To simplify synchronization, \gdmc\ implements a \emph{point-to-group sharing strategy}: the Master targets only one active worker for work-stealing at a time, tracked by the variable $\mathit{contacted}$.

While at least one worker is active, the Master monitors the cluster. 
If $\mathit{contacted} = \mathsf{nil}$, the Master selects a new target and sends a work request. 
Upon a positive non-blocking $\mathtt{isResponding}$ check, the Master facilitates direct transfers from the target to multiple idle workers (lines 11--13), assigning each sub-task a unique $idTask$ for aggregation.
This one-to-many transfer mitigates latency and rapidly re-integrates workers.
Finally, completed workers return to the $\mathit{idle}$ set; if the $\mathit{contacted}$ worker finished, the pointer is reset.

When all workers return to the idle state, the Master initiates result aggregation (line 18). Unlike \dmc{}, which streams partial expressions to the Master, \gdmc\ maintains stacks locally. To resolve a task, a worker must complete any previously shared nodes. When two workers communicate, the ``thief'' tracks the ``victim'' to return the computed value once the sub-stack for that $idTask$ is resolved. This propagates values back through the dependency chain until the root task ($idTask=0$) is completed.

\subsection{Worker-Master Integration}

We now detail the integration logic (the red block in Algorithm~\ref{alg:count}) that links the individual solver to the global orchestrator. To minimize communication overhead, the function $\mathsf{shouldPoll}$ acts as a throttle, ensuring the worker polls for Master requests only periodically (e.g., every $k=1,000$ decisions) to avoid system-call degradation.

The functions $\mathsf{shareRequest}$ and $\mathsf{transferWorkToThiefs}$ manage the worker's network interaction. $\mathsf{shareRequest}$ checks for an incoming $\mathtt{requestWork}$ signal. If detected, the worker begins a blocking communication to indicate if it has viable work. A worker may refuse to share if it is processing its last available formula or if the remaining sub-problems are too small to justify the communication cost. To prevent a ``ping-pong'' effect between active and idle states, we enforce a granularity threshold; for instance, formulas with fewer than 30 variables are solved locally rather than shared.

Finally, $\mathsf{transferWorkToThiefs}$ handles the offloading of a pending node to a thief. The victim performs a blocking send of the task data (variables and partial assignments) to the thief. Crucially, the victim does not wait for the result. It simply adds the assigned $idTask$ to the set of remote identifiers $\mathcal{I}$ of the current stack node (marking it as persistent) and immediately resumes its local search.

\subsection{Distributed Caching and Soundness}

Sharing arbitrary nodes from the stack is not permissible. 
As previously mentioned, sharing is restricted to nodes whose parents have already been shared. 
However, this topological constraint is insufficient to prevent \emph{cache inconsistency}, particularly when dealing with unsatisfiable sub-formulas.

Cache inconsistency occurs when the value $c$ stored in the cache for a formula $\Sigma$ satisfies $c \neq \norm{\Sigma}$. This phenomenon is a byproduct of conflict-driven clause learning. 
While learned clauses are treated as part of the formula for $\bcp$, they are not necessarily part of the original local sub-problem $\Sigma$. If a learned clause $\alpha$ is derived from variables in $\var{\Sigma}$ but is not a logical consequence of $\Sigma$ (i.e., $\Sigma \not\models \alpha$), using $\alpha$ to simplify $\Sigma$ can result in an incorrect model count.

This issue is well-documented in the context of component caching~\cite{DBLP:conf/sat/SangBBKP04}: when a sub-formula is found to be unsatisfiable, any cache entries added after the initial conflict must be invalidated. In a local solver, this cleaning is straightforward. In a distributed setting, however, sharing a sub-problem may prevent a worker from ever detecting the specific conflict that would have triggered a cache invalidation, leading to the propagation of ``poisoned'' results. 

\begin{example}[Example~\ref{ex:running} cont'ed]\label{ex:corrupt}
Suppose we first solve branch $\bar{x_7}$ completely (count 12), then explore $x_7$. Inside $x_7$, we branch on $x_1$. We share $\bar{x_1}$ as Task 1, and after decomposing the $x_1$ branch, we share component $\mathcal{C}_2$ as Task 2. 
Then
$\mathcal{S} = \big( 
    \langle \emptyset, 12, \emptyset, 0 \rangle, \allowbreak
    \langle \emptyset, 0, \{1\}, 0 \rangle, \allowbreak
    \langle \{\mathcal{C}_1\}, 1, \{2\}, 1 \rangle \big)$.

Assume that solving $\Sigma|_{\bar{x_7}}$ generated the learnt clause $\alpha = (\bar{x_1} \vee x_4)$. 
When solving the local component $\mathcal{C}_1$, this clause forces $x_4$ to be true, reducing $\mathcal{C}_1$'s count to 2 instead of the correct 3. 
If the remote component $\mathcal{C}_2 \models \bot$, the solver detects the conflict too late: the corrupt count for $\mathcal{C}_1$ (2) may already be computed and cached, causing incorrect results in future lookups, or worse, the inconsistency might never be detected.
\end{example}

We demonstrate that such inconsistencies only arise in branches containing an unsatisfiable connected component, thus ensuring soundness.

\begin{proposition}
    Let $\rho$ be the set of unit literals (the current path) such that $\Sigma$ is a connected component of $\Sigma_{\textit{root}}|_{\rho}$.
    Let $c$ be the model count computed for $\Sigma$. If every connected component of $\Sigma_{\textit{root}}|_{\rho}$ is satisfiable, then $c = \norm{\Sigma}$.
\end{proposition}

\begin{proof}
    Let $\Delta$ be the learned clauses. Assume for contradiction a cache inconsistency ($c \neq \norm{\Sigma}$) 
    occurs even if all components of $\Sigma_{\textit{root}}|_{\rho}$ are satisfiable. This implies the existence of $\alpha \in \Delta|_{\rho}$ where $\var{\alpha} \subseteq \var{\Sigma}$ but $\Sigma \not\models \alpha$.

    By construction, every clause learned is a logical consequence of the original formula, implying $\Sigma_{\textit{root}}|_{\rho} \models \alpha$. Since $\Sigma_{\textit{root}}|_{\rho}$ is partitioned into independent connected components $\{\Sigma, \mathcal{C}_1, \dots, \mathcal{C}_n\}$, the formula is logically equivalent to the conjunction of these components: $\Sigma_{\textit{root}}|_{\rho} \equiv \Sigma \wedge \mathcal{C}_1 \wedge \dots \wedge \mathcal{C}_n$. 
    
    If all components are satisfiable, $\alpha$ must be a logical consequence of $\Sigma$ alone. For $\alpha$ to be a consequence of the global formula without being a consequence of $\Sigma$, it would necessarily require the unsatisfiability of at least one other component $\mathcal{C}_i$ to ``force'' the implication through a contradiction. Since all components are satisfiable by hypothesis, $\Sigma \models \alpha$ must hold. This contradicts the assumption that $\Sigma \not\models \alpha$.
\end{proof}

Cache inconsistency has been addressed in distributed model counting. 
In \countAntom~\cite{DBLP:conf/sat/BurchardSB15}, the authors ignore the problem; however, their architecture differs. 
\countAntom{} manages a single global search tree at the Master, whereas our approach involves multiple workers independently constructing portions of the search space.
In \dmc{}, this issue does not arise because \d4 invokes a SAT solver at every decision node to ensure branch satisfiability.

However, most modern model counters forgo satisfiability checks at every decision node to avoid computational overhead. 
This design choice complicates the direct integration of our distributed framework with such solvers, as unverified branches can lead to cache inconsistencies. 
Fortunately, this can be mitigated by ensuring that all connected components of an \textsf{AND} node are verified as satisfiable before the node is marked as eligible for sharing. While this introduces a cost exceeding standard \bcp{}, the overhead remains negligible in practice because the solver is not invoked at every decision. Moreover, since work sharing is restricted to the upper levels of the search tree, the number of \textsf{AND} nodes requiring verification is relatively small. Consequently, the total number of calls to the SAT oracle is minimized, ensuring that these safety checks do not degrade global search performance.

\subsection{Worker Invocation}

Following the protocol established by \discount~\cite{DBLP:conf/kr/LagniezL25}, our framework invokes underlying model counters in \emph{assumptions mode}, a technique derived from incremental SAT solving~\cite{DBLP:conf/sat/NadelR12}. 
In this configuration, the sub-problem assigned to a worker is defined by a set of \emph{assumptions} (literals asserted only for the duration of a specific solver invocation). 
This approach is particularly advantageous for distributed model counting as it allows the underlying solver to retain its internal state across successive tasks. 
Since the base formula $\Sigma$ remains static, workers can preserve variable activity data and \emph{learnt clauses}, which remain logically valid for any sub-problem derived from $\Sigma$. 
Furthermore, maintaining a persistent component cache across multiple counting queries has been shown to be highly beneficial~\cite{DBLP:conf/aaai/LagniezM19}. 
By utilizing incremental invocations, \gdmc\ enables workers to reuse cache entries and variable weights, significantly accelerating the processing of related sub-problems.

Unlike distributed counters that treat each task as an isolated query, each worker in \gdmc\ manages a local collection of stacks. Upon receiving a sub-problem with a specific $idTask$, a worker initializes a corresponding stack to track the search and any internal decompositions. 
Rather than returning partial counts immediately to the Master, workers accumulate these solved sub-stacks locally to avoid frequent transfers of large arithmetic values. 
Once the Master determines that the entire search space has been covered (signaled by all workers entering the \emph{idle} state) it triggers the final reduction phase. 
During this phase, workers perform a decentralized reduction, propagating computed values back to the original victims via the stored $idTask$ dependencies. 
This strategy minimizes network congestion and ensures that the final model count for the entire formula ($idTask=0$) is only resolved once all dependencies are satisfied.

\section{Experiments~\label{empirical}}
Our framework is implemented in \texttt{C++} as an extension of \discount, leveraging concepts~\cite{DBLP:conf/popl/ReisS06} and templates.
We utilize concepts to define strict interfaces for the model counter, enabling seamless interchangeability of the stack manager and communication protocols.
\gdmc{} supports arbitrary-precision aggregation. Furthermore, this template-based architecture simplifies \textsf{MPI} integration by eliminating the need for invasive modifications to the solver's build configuration. The source code and experimental logs are available on Zenodo (\url{https://zenodo.org/records/20157902}).

We use \texttt{Glucose 3.0}~\cite{DBLP:conf/sat/AudemardLS13} in incremental mode and \texttt{FlowCutter}~\cite{DBLP:journals/jea/HamannS18} (PACE 2017) for tree decompositions (10s budget). Preprocessing utilizes \texttt{B+E}~\cite{DBLP:conf/ijcai/LagniezLM16} with the \texttt{equiv} option (vivification, backbone, occurrence elimination). To prevent worker idle time caused by slow, serial definability checks, we limit this phase to 5 s. Finally, we upgraded \texttt{SharpSAT-TD}~\cite{DBLP:conf/cp/KorhonenJ21} to support incremental assumptions.

\begin{figure}[t]
    \centering
 \includegraphics[width=0.45\textwidth]{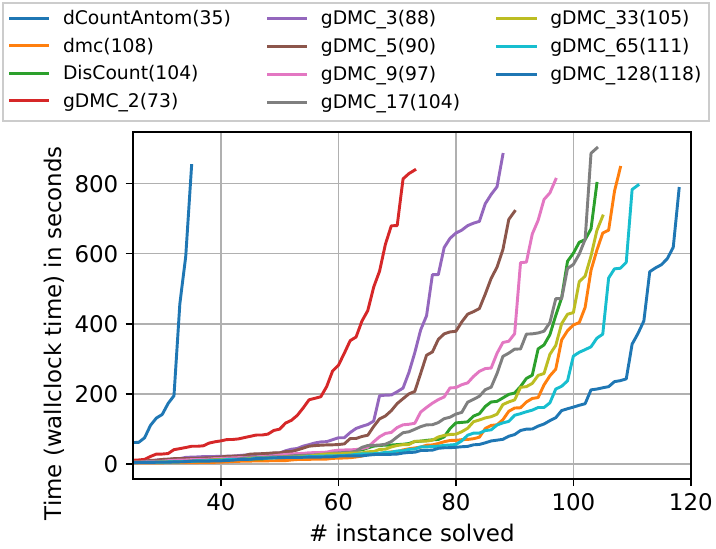}
\caption{\label{fig:cacuts}
Solved instances over time for \texttt{dCountAntom}, \texttt{dmc}, \discount, and \gdmc{} configurations. Parentheses in the legend indicate the total number of solved instances.
}
\end{figure}

We evaluated our approach against state-of-the-art distributed model counters \texttt{dCountAntom}~\cite{DBLP:conf/cluster/BurchardSB16}, \texttt{dmc}~\cite{DBLP:conf/ijcai/LagniezMS18}, and \discount~\cite{DBLP:conf/kr/LagniezL25}.
For a fair and accurate comparison, \discount~use \texttt{SharpSAT-TD} as the base solver.
For our experiments, work-stealing requests are issued periodically every $k=1,000$ decisions, and we do not share formulas with fewer than 30 variables. 
The benchmarks were taken from the most recent competition~\cite{Competition2021_23}. 

All experiments ran on a cluster of thirty-two Intel\textsuperscript{\textregistered} Xeon\textsuperscript{\textregistered} E5{-}2643 v4 CPUs at 3.30\,GHz, with Rocky Linux 9.5 (Linux kernel 5.14). Nodes are connected by 1\,GiB/s Ethernet. The software environment used GCC 11.5 and Open~MPI 5.1.0a1. Finally, a wall-clock time limit of 900 seconds and a memory limit of 32\,GiB were imposed on each run.

Figure~\ref{fig:cacuts} presents a cactus plot of solved instances over time for \dcountAntom, \dmc, \discount, and \gdmc. To assess the scalability of \gdmc, we evaluated configurations of $2, 3, 5, 9, 17, 33, 65$ and $128$ cores. Since \gdmc{} requires a dedicated master node, these configurations correspond to $1, 2, 4, \dots, 64$, and $127$ active workers, respectively. The competing solvers \dcountAntom, \dmc, and \discount~were executed using the full 128-core allocation.

The results demonstrate that increasing the core count significantly reduces the solving time for \gdmc{}, maximizing the number of instances solved within the limit. 
In the 128-core comparison, \dcountAntom{} is the least efficient (35 solved), followed by \discount{} (104) and \dmc{} (108), with \gdmc{} achieving the best performance in both count and speed.
Notably, contrary to findings in~\cite{DBLP:conf/kr/LagniezL25}, \discount{} trails slightly behind \dmc{} in this evaluation. 
We attribute this discrepancy to two distinct factors. 
First, the increased difficulty of this year's benchmarks hinders \discount's ability to effectively partition the search space into cubes. 
Second, \dmc{} uses double-precision arithmetic, thereby avoiding the computational overhead of the arbitrary-precision arithmetic employed by \discount{}.

\begin{figure}[t]
    \centering
\includegraphics[width=0.45\textwidth]{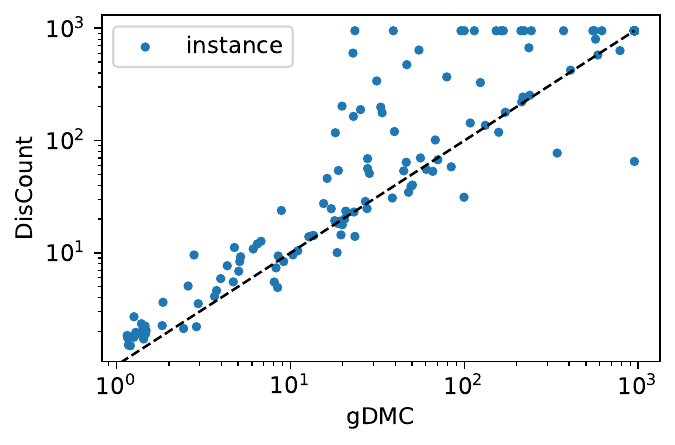}
\caption{\label{fig:scatter}
Scatter plot comparing solving times (in seconds) of \discount~($y$-axis) and \gdmc{} ($x$-axis) on a logarithmic scale. Each point corresponds to a single benchmark instance.}
\end{figure}

Figure~\ref{fig:scatter} compares solving times of \discount~and \gdmc{}; points above the diagonal favor \gdmc{}. \gdmc{} solves more instances, evidenced by top-border points (\discount~timeouts), and is consistently faster. This gap widens with instance difficulty. We attribute this efficiency to two factors: \gdmc{} avoids cube-generation overhead on simple instances, while dynamic work-stealing balances load on complex ones, mitigating the long-tail latency of static partitioning.


\begin{figure}[t]
    \centering
    \includegraphics[width=0.45\textwidth]{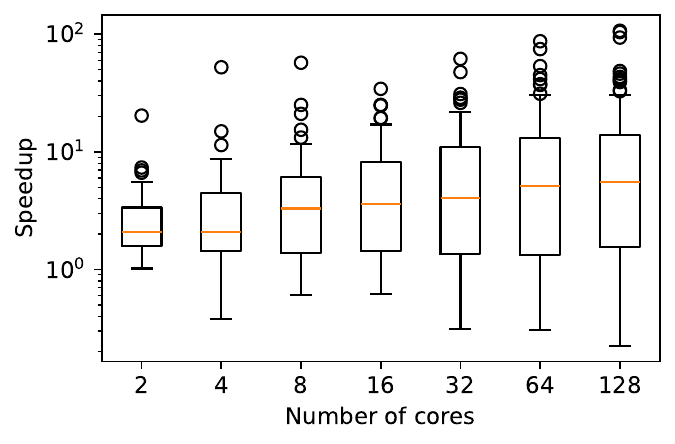}
    \caption{Distribution of speedups achieved by \gdmc{} across varying core counts, relative to the single-worker baseline (2 cores).\label{fig:speedup-boxplot}}
    
\end{figure}

Figure~\ref{fig:speedup-boxplot} depicts the distribution of speedups achieved by \gdmc{} as computational resources scale from 2 to 128 cores. 
The $y$-axis represents the speedup factor relative to a baseline configuration of 2 cores (1 master and 1 worker). 
For any instance solved in $t_n$ seconds using $n$ cores, the speedup is computed as $\frac{\min(t_{\text{2}}, 900)}{t_n}$, where $t_{\text{2}}$ is the baseline solving time. 
Since $t_{\text{2}}$ is capped at the 900-second timeout, these reported values constitute a conservative lower bound on the actual speedup for instances that time out in the baseline configuration.
The results demonstrate robust scalability.
We observe a consistent upward trend in the median speedup (indicated by the orange line) as the number of cores increases, confirming that \gdmc{} effectively utilizes additional worker nodes. 
Furthermore, the presence of high-performing outliers indicates that for many complex instances, the framework achieves near-linear speedups. 
This suggests that the dynamic work-stealing strategy successfully mitigates the idle time and communication overhead typically associated with distributed exact counting.

\section{Conclusion and Perspectives~\label{conclusion}}
In this paper, we presented \gdmc{}, a generic, solver-agnostic framework for DPLL-style distributed exact model counting. 
By leveraging C++ concepts and templates, we decoupled distributed orchestration from the core solving logic, enabling zero-overhead integration with state-of-the-art counters like \sharpsat{}. Our approach employs a novel \emph{point-to-group} work-stealing strategy and a local aggregation mechanism that ensures soundness and cache consistency without the bottleneck of global result streaming.

Future work focuses on three key directions. First, we plan to integrate probabilistic counters such as \ganak. Second, we aim to exploit our solver-agnostic design for DPLL-style counters to implement a heterogeneous portfolio, running distinct solvers concurrently to dynamically leverage their respective strengths on sub-problems. Finally, we intend to evaluate the scalability of our architecture on high-performance computing clusters with thousands of cores, testing the limits of our synchronization protocols in massively parallel environments.

\section*{Acknowledgments}
This work has been partly supported by the CERADOC project of the French National Agency for Research (ANR-25-CE23-3078), National Natural Science Foundation of China (No. 62532014), Jilin Province Science and Technology Department Project (20240602005RC), Scientific Research Project of the Education Department of Jilin Province (JJKH20250334KJ).

\bibliographystyle{named}
\bibliography{ijcai26}

\end{document}